\documentclass[thm-restate,a4paper,USenglish]{lipics-v2021}
\hideLIPIcs
\usepackage{amssymb}
\usepackage{amsmath}
\usepackage{mathtools}
\usepackage{amsthm}
\usepackage{xcolor}
\usepackage{cite}

\title{Impossibility Results for Strong
Linearizability:\\The Difficulty of Consistent Refereeing}

\titlerunning{Impossibility Results for Strong Linearizability}

\author{Hagit Attiya}{Technion - Israel Institute of Technology}{hagit@cs.technion.ac.il}{https://orcid.org/0000-0002-8017-6457}{Partially supported by the Israel Science Foundation (grant number 22/1425)}

\author{Armando Castañeda}{Universidad Nacional Autónoma de México (UNAM)}{armando.castaneda@im.unam.mx}{https://orcid.org/0000-0002-8017-8639}{Supported by the DGAPA PAPIIT project IN108723}

\author{Constantin Enea}{LIX, Ecole Polytechnique, CNRS and Institut Polytechnique de Paris}{cenea@irif.fr}{https://orcid.org/0000-0003-2727-8865}{Partially supported by the project SCEPROOF founded by the French ANR Agency and the NSF Agency from USA}

\authorrunning{Attiya, Castañeda and Enea}

\nolinenumbers

\begin{document}

\maketitle

\funding{
Hagit Attiya is supported by the Israel Science Foundation
(22/1425 and 25/1849).
Armando Castañeda is supported by the research projects DGAPA-PAPIIT IN108723
and IN103126, and SECIHTI CBF-2025-I-39.
Constantin Enea is partially supported by project SCEPROOF founded by
the French ANR Agency and the NSF Agency from USA.}

\abstract{
This paper studies the relation between agreement
and strongly linearizable implementations of various objects.
This leads to new results about implementations of concurrent
objects from various primitives including \emph{window registers},
\emph{interfering primitives} and \emph{stacks}.

We identify that lock-free, respectively, wait-free, strongly linearizable
implementations of several concurrent objects entail
a form of agreement that does not require full consensus but cannot be 
implemented in a strongly-linearizable manner with combinations of non-universal primitives.
In both cases, lock-free and wait-free, this form of agreement requires
a distinguished process to \emph{referee} a \emph{competition}
that involves all other processes. Our results
show that \emph{consistent} refereeing of such competitions
(i.e., once resolved, the outcome cannot be revised in any extension of the execution)
requires high coordination power.

More specifically, two \emph{contest} objects are defined and used
to help characterize coordination constraints imposed by strong linearizability
in lock-free and wait-free
implementations, respectively.
Both objects are strictly weaker than consensus,
in the sense that they have a wait-free linearizable
(in fact, decisively linearizable)
implementation from reads and writes.

The contest objects capture strong linearizability in two complementary ways. 
First, they admit
strongly linearizable implementations from several high-level objects such as queues, snapshots,
and counters, so impossibility results for the contest objects immediately carry over to these
objects. 
Second, they admit powerful impossibility results for strong
linearizability that involve \emph{window registers}, 
\emph{interfering primitives} and \emph{stacks}, which are non-universal.
}

\section{Introduction}

A key way to construct complex distributed systems is through modular
composition of linearizable concurrent objects~\cite{HerlihyW1990}.
Yet linearizable objects do not always compose correctly
with randomized programs~\cite{HadzilacosHT2020,GolabHW2011},
or with programs that should not leak information~\cite{AttiyaE2019}.
This deficiency is addressed by \emph{strong linearizability}~\cite{GolabHW2011},
a restriction of linearizability,
which ensures that properties holding when a concurrent program
is executed in conjunction with an atomic object,
continue to hold when the program is executed
with a strongly linearizable implementation of the object.
Strong linearizability was shown~\cite{AttiyaE2019,DongolSW2023}
to preserve \emph{hyperproperties}~\cite{ClarksonS2010},
such as security properties and
probability distributions of reaching particular program states.

The \emph{universal} implementation of concurrent objects
from compare\&swap~\cite{H91},
is strongly linearizable~\cite{GolabHW2011}.
But any attempt to reduce the usage of costly, hardware-dependent compare\&swap
has been shown to be challenging and require subtle algorithms 
(see, for example,~\cite{HwangW2021}).

There have been several impossibility results showing there are no
strongly linearizable implementations of common objects
from widely-available primitives, like read, write, swap and test\&set,
even when they have linearizable implementations from the same primitives.
For example, max registers, snapshots, or monotonic counters
do not have \emph{wait-free} strongly linearizable implementations,
even with multi-writer registers~\cite{DenysyukW2015}.
Single-writer registers do not suffice even for \emph{lock-free}
strongly linearizable multi-writer registers, max registers,
snapshots, or counters~\cite{HelmiHW2012}.
While queues and stacks admit lock-free and wait-free linearizable
implementations from read, write, test\&set and swap~\cite{AGM07,HerlihyW1990},
there are no lock-free strongly linearizable implementations of those objects
from the same primitives~\cite{AttiyaCE2024}.

Strong linearizability seems to have a close connection with the
cornerstone \emph{consensus} task.
Indeed, most of the impossibility results rely on \emph{valency}
arguments~\cite[Chapter 7]{AttiyaE2014},
heavily used to study agreement tasks, in particular, consensus.
These proofs are ad-hoc and, in some cases, quite complicated.
Other proofs~\cite{AttiyaCE2024} are by reduction from agreement tasks,
like \emph{set agreement}, a generalization of consensus,
but they rely on the non-trivial assumption that the underlying shared
memory base-objects or primitives are readable.\footnote{Readability might change the coordination
power of a high-level object or primitive. For example, queues and stacks, whose consensus
numbers are both two, have infinite consensus number when read operations are added.}

\paragraph*{Our results and their ramifications}

This paper aims to elucidate the relation between agreement
and strongly linearizable implementations of various objects.
By doing so, we are able to derive a host of new results about
implementations of concurrent objects from various primitives,
including in particular, \emph{$w$-window registers}~\cite{MostefaouiPR18},
whose consensus number is exactly $w$,
and atomic \emph{stacks}, whose consensus number is $2$,
that were not considered in prior work.

We identify that lock-free and wait-free strongly linearizable
implementations of several concurrent objects entail
a form of agreement that is weaker than consensus but impossible
to strongly linearizable implement with combinations of non-universal primitives.
In both cases, lock-free and wait-free, this form of agreement requires
a distinguished process to \emph{referee} a \emph{competition}
that involves all other processes. Our main results, explained below,
show that \emph{consistent} refereeing of such competitions
(i.e. the outcome of the competition does not change in extensions of the current execution,
as in strongly linearizable implementations)
requires high coordination power.

More in detail, we define two
\emph{contest} objects which are strictly weaker than consensus and
helps to derive impossibility results for
lock-free or wait-free strongly linearizable implementations of several concurrent objects.
These objects are strictly weaker than consensus, since, unlike consensus,
they have a wait-free linearizable implementation from reads and writes.
The implementations are in fact decisive linearizable~\cite{BSLS2024}.
Furthermore, the objects admit
strongly linearizable implementations from many ``high-level'' objects like queues, snapshots, counters, etc, and therefore, impossibility results about
them have many interesting corollaries and
also admit powerful impossibility results for strong linearizability that involve stacks, window registers and many non-universal synchronization primitives.

For impossibility results that concern lock-freedom, we define a \emph{contest} object in which,
similarly to \emph{id-consensus}~\cite{Chandra1996},
the id of one among several \emph{competitor} processes is picked.
However, in the contest object a single, non-competing process
is the \emph{referee} who picks the winner and is the \emph{only one knowing this information}.
Unlike agreement tasks, competitor processes in a contest object
do not need to agree on an outcome and always return the same response;
only a single distinguished process observes the result of the competition.
Indeed, the contest object admits a decisive linearizable implementation from reads and writes
(Theorem~\ref{thm:contest-from-read-write}).

Our first main contribution is that there is no \emph{lock-free strongly linearizable}
implementation for~$n$ processes of a contest object from \emph{atomic stacks, $(n-2)$-window
registers and interfering primitives} (Theorem~\ref{thm:no contest from window}).
The class of interfering primitives~\cite{H91} has consensus number two,
and includes read, write, test\&set, swap, among other useful primitives.
The proof relies on a simple valency argument.
It shows that even relatively strong but non-universal synchronization primitives
do not suffice for lock-free strongly linearizable contest implementations.
We complement this impossibility result by presenting a wait-free strongly
linearizable implementation of a contest object from $(n-1)$-window registers
(Theorem~\ref{thm:contest-from-window}).

Since a queue can be used in a wait-free strongly linearizable
implementation of a contest object,
it follows that there is no lock-free strongly linearizable implementation
of a queue from stacks, $(n-2)$-window registers and interfering primitives
(Corollary~\ref{cor:lf sl queue and stack}).

Some very useful objects, including counters, snapshots and max registers,
have strongly linearizable implementations for a bounded number of invocations,
which are also \emph{lock-free}.
To show that they do not have \emph{wait-free} strongly linearizable
implementations for an unbounded number of operations,
we define the \emph{long-lived} contest object.
In this object, competitor processes can participate any number of times in the competition,
and the referee must return an integer $x$ such that every competitor
has participated \emph{no more} than $x$ times so far. 
Again, this information is known only by the referee.
The long-lived contest object has a wait-free decisively linearizable
implementation from reads and writes (Theorem~\ref{thm:ll-contest-from-read-write}).

Using a group valency argument~\cite{DenysyukW2015},
we prove that there is no wait-free strongly linearizable implementation
of a long-lived contest object from reads and writes only, for $n \geq 3$
(Theorem~\ref{thm:no ll contest from read write}).
We extend this result also to implementations from $w$-window registers
and test\&set objects, for $n \geq 4$ and $3w \leq n$
  (Theorem~\ref{thm:ll-contest-from-window-ts}).
Since either a counter or a snapshot or a max register or a fetch\&increment
or a fetch\&add can be used in a
  wait-free strongly linearizable implementation of a contest object,
  this also means there are no wait-free strongly linearizable implementations
  of these objects from the same primitives
  (Corollary~\ref{cor:no wf sl counter from rw} and
  Corollary~\ref{cor:no wf sl counter from window}).

\paragraph*{Additional related work}

If one only requires \emph{obstruction-freedom},
which ensures an operation to complete only if it executes alone,
then any strongly linearizable object can be implemented
from single-writer registers~\cite{HelmiHW2012}.

When considering the stronger property of \emph{lock-freedom},
which requires that as long as some operation is pending, some
operation completes, single-writer registers are not sufficient
for implementing strongly linearizable multi-writer registers,
max registers, snapshots, or counters~\cite{HelmiHW2012}.
If the implementations can use multi-writer registers, though, it
is possible to get lock-free implementations of max registers,
snapshots, and monotonic counters~\cite{DenysyukW2015}, as well as of
objects whose operations commute or overwrite~\cite{OvensW19}.

It was also shown~\cite{AttiyaCH2018} that there is no
lock-free implementation of a queue or a stack with universal helping,
from objects whose readable versions have consensus number less than the
number of processes, e.g., readable test\&set.
The proof is by reduction to consensus, which is possible due to the readability assumption.
Universal helping is a formalization of helping mechanisms used
in universal constructions, and that can be understood as an
eventual version of strong linearizability.
Using a similar reduction, it is shown that lock-free strongly linearizable queues and
stacks are impossible from interfering primitives~\cite{AttiyaCE2024}.
Again, the reduction is based on the assumption that base objects are readable.
Readability does not change the consensus number of interfering primitives.
Through a valency argument (instead of a reduction),
we extend this impossibility result for queues to encompass $(n-2)$-window registers and stacks
(Corollary~\ref{cor:lf sl queue and stack}),
and remove the readability assumption.

For the even stronger property of {\em wait-freedom}, which requires every
operation to complete, it is possible to implement strongly linearizable
bounded max registers
from multi-writer registers~\cite{HelmiHW2012}, but it is impossible
to implement strongly linearizable max registers, snapshots, or monotonic
counters~\cite{DenysyukW2015} even with multi-writer registers.
The impossibility in~\cite{DenysyukW2015}, based on group valency arguments,
is for counters and then extended, via a reduction, to max registers and snapshots.
Our impossibility is inspired by the group valency argument used in~\cite{DenysyukW2015},
and for the case of four or more processes,
it strengthens the result to encompass test\&set and window registers.

It has been shown that, for any number of processes,
there are wait-free linearizable implementations
of fetch\&increment and fetch\&add from reads, writes and test\&set~\cite{AfekWW1993,Afek1999W}.
Our results show that this is no longer true if we consider strong linearizability instead,
and even if $w$-window registers are available, with $3w \leq n$.

\section{Model of Computation}

We consider a standard shared memory system with $n$ asynchronous processes, $p_0,\ldots,p_{n-1}$,
which may crash at any time during an execution.
Processes communicate with each other by applying \emph{atomic}
operations to shared \emph{base objects}.

A \emph{(high-level) concurrent object}
is defined by a state machine consisting of a set of states,
a set of operations, and a set of transitions between states.
Such a specification is known as \emph{sequential}.
An \emph{implementation} of an object $T$ is a distributed
algorithm $\mathcal{A}$ consisting of a local state machine
$\mathcal{A}_p$, for each process $p$.
$\mathcal{A}_p$ specifies which primitive operations on base objects $p$ applies
and which local computations $p$ performs in order to return a response
when it invokes an operation of $T$.
Each of these base object operation invocations and local computations is a \emph{step}.
For the rest of this section,
fix an implementation $\mathcal{A}$ of an object $T$.

We use $w$-window registers,
whose consensus number is exactly~$w$~\cite{MostefaouiPR18}.
Specifically,
a $w$-\emph{window register}~\cite{MostefaouiPR18} stores the sequence
of the last $w$ values  written to it (or the last $x$ values
when only  $x < w$ values have been written).
A $write$ with input $v$ appends $v$ at the end of the sequence and
removes the first one if the sequence is already of size $w$.
A $read$ operation returns the current sequence.
A standard register is a $1$-window register.
\emph{Interfering} primitives~\cite{H91} include read, write, test\&set,
swap and fetch\&add operations, among others.
More specifically, at any state of a register,
application of any pair of interfering primitives either commute or one overwrites the other.

A \emph{configuration} $C$ of the system contains the states of
all shared base objects and processes.
In an \emph{initial} configuration,
base objects and processes are in their initial states.
Given a configuration $C$ and a process $p$, $p(C)$ is the configuration
after $p$ takes its next step in~$C$.
Moreover, $p^{0}(C) = C$ and for every $n\in \mathbb{N}$,
$p^{n+1}(C) = p(p^{n}(C))$.
Note that the next step $p$ takes in $C$ depends only on
its local state in $C$.

An \emph{execution of $\mathcal{A}$ starting from $C_0$}
is a (possibly infinite) sequence
$C_0 e_1 C_1 e_2 C_2 \cdots$, where each $e_i$ is a step of a process,
or an invocation/response of a high-level operation by a process
and if $e_i$ is a step, then $C_{i+1} = e_i(C_i)$;
furthermore, the sequence satisfies:
\begin{enumerate}
\item Each process can invoke a new (high-level) operation only
 when its previous operation (if there is any) has a corresponding response,
 i.e., executions are \emph{well-formed}.
\item A process takes steps only between an invocation and a response.
\item For any invocation of process $p$,
  the steps of $p$ between that invocation and the following response of $p$,
  if there is one,
  correspond to steps of $p$ that are specified by $\mathcal{A}_p$.
\end{enumerate}

An execution $\beta$ is an \emph{extension} of a finite execution $\alpha$
if $\alpha$ is a prefix of $\beta$.
A configuration $C$ is \emph{reachable} if there is a finite execution~$\alpha$ 
starting from an initial configuration whose last configuration is $C$;
we say that $\alpha$ \emph{ends} with $C$.
A configuration $C'$ is \emph{reachable} from a configuration $C$
if there is a finite execution starting from $C$ that ends with $C'$.
Two configurations $C$ and $C'$ are \emph{indistinguishable} to process $p$
if the state of every base object and the state of $p$
are the same in $C$ and $C'$.

An operation in an execution is \emph{complete} if both its invocation
and response appear in the execution.
An operation is \emph{pending} if only its invocation appears in the execution.
An implementation is \emph{wait-free} if every process completes each of
its operations in a finite number of its steps;
formally, if a process executes infinitely many steps in an execution,
it completes infinitely many operations.
An implementation is \emph{lock-free} if whenever processes
execute steps, at least one of the operations terminates;
formally, in every infinite execution,
infinitely many operations are complete.
Thus, a wait-free implementation is lock-free
but not necessarily vice versa.

In an execution, an operation $op$ \emph{precedes}
another operation $op'$ if
the response of $op$ appears before the invocation of $op'$.

\emph{Linearizability}~\cite{HerlihyW1990} is the standard notion used
to identify a correct implementation.
Roughly speaking, an implementation is linearizable if each operation
appears to take effect \emph{atomically} at some time between its invocation
and response, hence operations' real-time order is maintained.
Formally, let $\mathcal{A}$ be an implementation of an object $T$.
An execution $\alpha$ of $\mathcal{A}$ is \emph{linearizable} if
there is a sequential execution $S$ of $T$
(i.e., a sequence of matching invocation-response pairs, starting with an invocation)
such that:
\begin{enumerate}
\item
$S$ contains every complete operation in $\alpha$ and
some of the pending operations in $\alpha$.
Hence, the output values in the matching responses of an invocation in
$S$ and the complete operations in $\alpha$ are the same.
\item
If $op$ precedes $op'$ in $\alpha$,
then $op$ precedes $op'$ in $S$;
namely, $S$ respects the \emph{real-time} order in~$\alpha$.
\end{enumerate}
$\mathcal{A}$ is \emph{linearizable} if all its executions are linearizable.

An implementation of a data type is \emph{strongly
linearizable}~\cite{GolabHW2011} if linearizations of executions
can be defined by appending operations to linearizations of prefixes.
Formally,
there is a \emph{prefix-closed} function $L$ mapping each execution
to a linearization such that
if an execution $\alpha$ is a prefix of an execution $\beta$,
then $L(\alpha)$ is a prefix of $L(\beta)$. \emph{Decisive linearizability}~\cite{BSLS2024} is a weakening of strong linearizability which requires that there is a function $L$ mapping each execution
to a linearization, such that if an execution $\alpha$ is a prefix of an execution $\beta$, then $L(\alpha)$ is a \emph{subsequence} of $L(\beta)$.\footnote{A sequence $\sigma$ is a sub-sequence of another sequence $\sigma'$ if $\sigma$ can be obtained from $\sigma'$ by deleting some symbols.}

\section{Motivating Example: Impossibility of Lock-Free
Strongly Linearizable Queues from Interfering Primitives}

We start with a crisp example that demonstrates how strongly linearizable
objects can be used to pick a decision among a set of competing values,
and that this property can be used to show they cannot be implemented
from interfering primitives.

\begin{theorem}
There is no lock-free strongly linearizable queue implementation
from interfering primitives, for {four} or more processes.
\end{theorem}

\begin{proof}
Towards a contradiction, suppose that there is such an implementation $A$.
Let us consider all \emph{failure-free} executions where
{process $p_i$, $i \in \{0,1,2\}$, executes a single $enqueue(i)$ operation,
and~$p_3$ executes a single $dequeue()$ operation}
starting once \emph{at least} one enqueue operation is completed.
By lock-freedom of $A$, eventually the operation of $p_3$ starts and completes.
By the linearizability of $A$, it is not possible that $p_3$'s dequeue operation returns empty.

A configuration $C$ is \emph{multivalent} if for distinct values $v,v' \in \{0,1,2\}$,
there are extensions of $C$ in which $p_3$'s dequeue returns $v$ and $v'$, respectively.
A configuration $C$ is $v$-\emph{univalent}, {$v \in \{0,1,2\}$},
if in all extensions of $C$ where $p_3$'s dequeue returns a value,
the value is $v$.
Recall that in both cases we focus on executions where $p_3$'s dequeue 
starts after at least one enqueue is completed.

For {$i \in \{0,1,2\}$}, consider the following execution $\alpha_i$:
process $p_i$ executes until completion $enqueue(i)$,
and then $p_3$ executes $dequeue()$ until completion.
Execution $\alpha_i$ exists due to lock-freedom of $A$,
and the dequeue operation of $p_3$ returns $i$,
due to linearizability of $A$.
This implies that the initial configuration $C_0$ is multivalent.

Since the initial configuration is multivalent, it follows that
every reachable configuration is either 0-univalent, 1-univalent,
2-univalent or multivalent;
these properties are mutually exclusive.

It is obvious that if $p_3$ completes its dequeue, then the configuration is univalent.
Interestingly, with strong linearizability the same holds also when one of the other
process completes its enqueue operation.
Note that this is the only place where strong linearizability is used.

\begin{claim}
\label{claim-queue}
A reachable configuration in which the enqueue of {$p_0, p_1$ or $p_2$} is complete
is univalent.
\end{claim}

\begin{proof}
Consider a prefix-closed linearization function $L$ of $A$,
which is guaranteed to exist due to strong linearizability of $A$.
Let $\alpha$ be an execution that ends with configuration $C$,
in which at least one enqueue operation completes.
We observe the following cases:
\begin{itemize}

\item The dequeue operation of $p_3$, denoted $op$, appears in $L(\alpha)$ returning
value {$v \in \{0,1,2\}$}.
If $op$ is complete then clearly $v$ is the non-empty value that $p_3$
returns in every extension of $C$.
Thus $C$ is univalent.
If $op$ is pending, then consider an extension $C'$ of $C$,
and let $\alpha'$ be an execution that extends $\alpha$ and ends at $C'$.
Since $L$ is prefix-closed, $L(\alpha)$ is prefix of $L(\alpha \cdot \alpha')$,
hence $op$ appears in $L(\alpha \cdot \alpha')$ returning $v$. Therefore, if $op$ is completed in $C'$,
it returns $v$. Thus $C$ is univalent.

\item The dequeue operation of $p_3$ does not appear in $L(\alpha)$.
Let $v$ the value enqueued by the first enqueue operation in $L(\alpha)$.
Consider any extension $C'$ of $C$ such that the dequeue operation of
$p_3$ completes, returning $u$, and let $\alpha'$ be an execution that extends $\alpha$ and ends at $C'$.
Since $L$ is prefix-closed, $L(\alpha)$ is prefix of $L(\alpha \cdot \alpha')$,
hence $v$ is the first enqueue value in $L(\alpha \cdot \alpha')$,
from which follows that $u=v$, due to linearizability.
Thus $C$ is univalent. \qedhere
\end{itemize}

\end{proof}

Starting with the initial configuration $C_0$, which is multivalent,
we let {$p_0, p_1$ and $p_2$} take steps as long as the next configuration is multivalent.
This execution cannot go forever since the operations of {$p_0, p_1$ and $p_2$}
eventually terminate, due to lock-freedom,
and every configuration where an enqueue is terminated is univalent, due to the claim above.
Thus, we reach a \emph{critical} configuration $C$ such that
(1) {$p_0, p_1$ and $p_2$} are active,
(2) $p_3$ is idle,
(3) $C$ is multivalent, and
(4) a step of one of {distinct processes $p$ and $q$ makes the configuration $v$-univalent,
and a step of the other makes the configuration $u$-univalent, with $u \neq v$.
Without loss of generality, let $p$ and $q$ be processes $p_0$ and~$p_1$.}

Let $R_i$, $i \in \{0,1\}$, be the base object $p_i$ is about to access
with primitive $f_i$, in the configuration~$C$.
We complete the proof by considering the following cases:
\begin{itemize}

\item $R_0 \neq R_1$. {This case leads to a contradiction because
the configurations $p_1 p_0 (C)$ and $p_0 p_1 (C)$ have different valencies, and
both are indistinguishable to $p_2$ and $p_3$, which implies that
$p_3$ cannot return distinct values in the extensions of $p_1 p_0 (C)$ and $p_0 p_1 (C)$
in which first $p_2$ runs solo until it completes its enqueue, and then $p_3$ starts and completes its dequeue.}

\item $R_0 = R_1$ and $f_0$ and $f_1$ commute. We reach a contradiction as well,
because $p_1 p_0 (C)$ and $p_0 p_1 (C)$ are indistinguishable to {$p_2$ and $p_3$}.

\item $R_0 = R_1$, and without loss of generality,
suppose that $f_0$ overwrites $f_1$.
Therefore, configurations $p_0(C)$ and $p_0 p_1(C)$ are indistinguishable to {$p_2$ and $p_3$},
which is a contradiction. \qedhere
\end{itemize}
\end{proof}

\section{Impossibility of Lock-Free Strongly Linearizable Implementations: The One-Shot Contest Object}
\label{sec:lock-free}

We extend the example of a queue by considering a contest object, defined next.
The \emph{contest} sequential, one-shot object provides two operations,
$compete()$ that can be invoked by processes $p_1, \hdots, p_{n-1}$, the \emph{competitors},
and $decide()$ that can be invoked by process $p_0$, the \emph{referee}.
The specification is that $compete()$ always returns $true$, and $decide()$ returns the index
of the first process that executed $compete()$, or $false$ if there is no such operation.
The contest object is weaker than a consensus object,
as it can be implemented from reads and writes.

\begin{restatable}{theorem}{contestfromreadwrite}
\label{thm:contest-from-read-write}
There is a wait-free (decisively) linearizable contest implementation from
reads and writes, for $n \geq 2$ processes.
\end{restatable}

\begin{proof}
Consider the following read/write wait-free implementation.
The processes share a variable~$X$, initialized to $false$.
Every competitor $p_i$ reads $X$ and if the read value is $false$,
it writes $i$ in $X$ and returns $true$, otherwise it returns $true$ without writing to $X$.
The referee returns the value it reads from $X$.

We argue that the implementation is decisively linearizable.
Consider any finite execution $\alpha$ of the algorithm.
Pending operations of $\alpha$ are treated as follows.
For any pending $compete()$ operation in $\alpha$, it is completed if it writes to $X$,
otherwise it is removed.
If $\alpha$ has a pending $decide()$ operation, this operation is removed.
For simplicity, let $\alpha$ itself denote the modified execution.
Reading $\alpha$ from beginning to end, operations are linearized as follows:

\begin{itemize}
\item if the current step is a write to $X$ from a $compete()$ operation,
there are two sub-cases: if there is no $decide()$ in the linearization so far, the $compete()$ operation
is linearized at the first position of the linearization, otherwise
it is appended to the linearization.

\item if the current step is a read of $X$ returning a process index from a $compete()$ operation, then this is appended to the linearization (this $compete()$ operation will not write to $X$)

\item if the current step is a read of $X$ from a $decide()$ operation, then this is appended to the linearization,
\end{itemize}

Observe that all $compete()$ operations that write to $X$ are concurrent.
Thus,  the linearization just defined respects the real-time order of operations in $\alpha$.
The linearization is valid for the contest object because,
as operations are linearized reading $\alpha$ from beginning to end,
at all times $X$ contains the index of the competitor whose $compete()$
appears first in the linearization, or $false$ if there is no such operation,
which makes the response of $decide()$ correct.

This is decisively linearizable because the order between operations 
already linearized in $\alpha$ does not change as $\alpha$ is extended with more steps.
\end{proof}

A strongly linearizable contest object can be implemented from $(n-1)$-window registers.

\begin{restatable}{theorem}{contestfromwindow}
\label{thm:contest-from-window}
There is a wait-free strongly linearizable contest implementation from
$(n-1)$-window registers, for $n \geq 2$ processes.
\end{restatable}

\begin{proof}
Consider an implementation where the processes share an $(n-1)$-window register $X$, initialized to $\bot$,
and when a competitor $p_i$ invokes $compete()$, writes $i$ in $X$,
and when the referee invokes $decide()$, 
it reads $X$ and returns the index of the first process in the read value,
or $false$ if the read value is $\bot$.

The implementation is obviously wait-free.
To prove strong linearizability, in any execution $decide()$ operations are linearized
at the step where the referee reads $X$, and a $compete()$ operation is linearized at the
step where the competitor writes in $X$.
Strong linearizability follows from the observation that there are $n-1$ competitors
and every competitor writes at most once in $X$,
and the assumption that $X$ is atomic.
\end{proof}

In contrast, $(n-2)$-window registers do not suffice to
implement a strongly linearizable contest object,
even together with interfering primitives.

\begin{theorem}
\label{thm:no contest from window}
There is no wait-free strongly linearizable contest implementation
from interfering primitives and $(n-2)$-window registers and atomic stacks,
for {$n \geq 4$ processes}.
\end{theorem}

\begin{proof}
By contradiction, suppose that there is such an implementation $A$.
We will reason about all \emph{failure-free} executions where each competitor
process $p_1, \hdots, p_{n-1}$ executes its $compete()$ operation,
and the referee $p_0$ starts its $decide()$ operation
once \emph{at least} one $compete()$ operation is completed.

\begin{restatable}{claim}{claimone}
$p_0$ eventually completes its operation, obtaining a non-$false$ value.
\end{restatable}

\begin{proof}
Since $p_1,\ldots,p_{n-1}$ are correct and $A$ is wait-free,
eventually their $compete()$ operations complete.
Thus, eventually the operation of $p_0$ starts and completes,
and due to linearizability of $A$,
the operation of $p_0$ returns the index of a process.
\end{proof}

A configuration $C$ is \emph{multivalent} if there are two distinct values
$v, v' \in \{1,\ldots,n-1\}$,
such that there are two extensions from $C$ in which
$p_0$ returns $v$ and $p_0$ returns $v'$, respectively.
A configuration $C$ is~$v$-\emph{univalent}, $v \in \{1,\ldots,n-1\}$,
if in all extensions from $C$
where $p_0$ returns a value, the value is~$v$.
Recall that in both definitions we focus on executions where the referee starts once at least one competitor completes.

\begin{restatable}{claim}{claimtwo}
The initial configuration $C_0$ is multivalent.
\end{restatable}

\begin{proof}
For $i \in \{1,\ldots,n-1\}$, consider the following execution $\alpha_i$:
process $p_i$ executes until completion $compete()$,
and then $p_0$ executes $decide()$ until completion.
Execution $\alpha_i$ exists as $A$ is wait-free,
and the $decide()$ operation of $p_0$ returns $i$,
due to linearizability of $A$.
\end{proof}

The next claim follows from the definitions
and the fact that the initial configuration is multivalent.

\begin{restatable}{claim}{claimthree}
A reachable configuration is either $v$-univalent,
for some $v \in \{1, \ldots, n-1\}$, or multivalent;
furthermore, the properties are mutually exclusive.
\end{restatable}

It is obvious that if $p_0$ completes its $decide()$ operation in a configuration $C$,
then $C$ is univalent.
Interestingly, with strong linearizability the same holds also when
one of the competitors complete its $compete()$ operation.
The proof of the next claim is similar to the proof of Claim~\ref{claim-queue}
above.

\begin{restatable}{claim}{claimfour}
Any reachable configuration where some $compete()$ operation completes
is univalent.
\end{restatable}

\begin{proof}
Consider any prefix-closed linearization function $L$ of $A$,
which is guaranteed to exist due to strong linearizability of $A$.
Let $C$ be any reachable configuration with at least one completed
$compete()$ operation, and let $\alpha$ be an execution that ends at $C$.
We observe the following cases:
\begin{itemize}

\item The $decide()$ operation of $p_0$ appears in $L(\alpha)$ returning value $v \, (\neq false)$.
Such operation of $p_0$, denoted $op$, might be pending or complete in $C$.
If $op$ is complete then clearly $v$ is the non-$false$ value
that $p_0$ returns in every extension of $C$, as $op$ is already complete.
Thus, $C$ is $v$-univalent.
If $op$ is pending, then consider an extension $C'$ of $C$,
and let $\alpha'$ be an execution that extends $\alpha$
and ends at $C'$.
Since $L$ is prefix-closed, $L(\alpha)$ is prefix of $L(\alpha \cdot \alpha')$,
hence $op'$ appears in $L(\alpha \cdot \alpha')$ returning $v$.
Therefore, if $op$ is completed in $C'$,
it returns $v$. Thus, $C$ is $v$-univalent.

\item The $decide()$ operation of $p_0$ does not appear in $L(\alpha)$.
Let $v$ the process index of the first $compete()$ operation in $L(\alpha)$,
whose existence is guaranteed by the hypothesis of the lemma and since $L$ is a linearization function.
Consider any extension $C'$ of $C$ such that the $decide()$ operation of
$p_0$ returns value $u \, (\neq false)$,
and let $\alpha'$ be an execution that extends $\alpha$ and ends at $C'$.
Since $L$ is prefix-closed, $L(\alpha)$ is prefix of $L(\alpha \cdot \alpha')$,
hence $v$ is the process index of the first $compete()$ in $L(\alpha \cdot \alpha')$,
from which follows that $u=v$, due to linearizability.
Thus, $C$ is $v$-univalent. \qedhere
\end{itemize}
\end{proof}

We now prove the existence of a reachable configuration $C$
that is \emph{critical}, that is:
\begin{claim}
\label{claim:critical}
There is a reachable configuration $C$ such that
\begin{enumerate}
\item
(1)~$p_1, \ldots, p_{n-1}$ are active,
\item
(2)~$p_0$ is idle,
\item
(3)~$C$ is multivalent,
\item
(4)~the step of any $p_i$, $1 \leq i \leq n-1$, makes the next configuration $v_i$-univalent,
\item
(5)~for some processes $p_i$ and $p_j$, $v_i \neq v_j$.
\end{enumerate}
\end{claim}

\begin{proof}
We know that $C_0$ is multivalent.
Starting from $C_0$, we let processes
$p_1,\ldots,p_{n-1}$ take steps as long as
the next configuration is multivalent.
This cannot go forever since the operation of some process
eventually terminates, due to wait-freedom, and every configuration
where a $compete()$ of $p_i$ terminates is univalent.
Thus, we can reach a multivalent configuration $C$
where a step of each process $p_1, \ldots, p_{n-1}$
makes the next configuration univalent.
Let $v_i$ the value the configuration becomes $v_i$-univalent after the step of~$p_i$.
Since $C$ is multivalent,
it cannot be that $v_1 = \hdots = v_{n-1}$.
\end{proof}

We can now complete the proof of the theorem.
Consider the critical configuration $C$ 
and let $R_i$ and $R_j$ be the base objects $p_i$ and $p_j$ are
about to access in $C$, through primitives $f_i$ and $f_j$.
We complete the proof by showing a contradiction in each of the following cases,
{recalling that there are at least three competitors (since $n \geq 4$)}:
\begin{enumerate}
\item $R_i \neq R_j$. This case leads to a contradiction because configuration
$p_j p_i (C)$ is $v_i$-valent and configuration $p_i p_j (C)$ is $v_j$-valent, but
$p_j p_i (C)$ and $p_i p_j (C)$ are indistinguishable to $p_0$
{and all competitors distinct to $p_i$ and $p_j$, which implies that $p_0$ cannot output distinct values
in the extensions of $p_j p_i (C)$ and $p_i p_j (C)$ in which a competitor $p_k \neq p_i,p_j$
runs solo and completes its $compete()$ operation, and then $p_0$ starts and completes its $decide()$ operation.}

\item $R_i = R_j$ and is an interfering object:
\begin{enumerate}
\item $f_i$ and $f_j$ commute. We reach a contradiction too because
$p_j p_i (C)$ and $p_i p_j (C)$ are indistinguishable to $p_0$
{and all competitors distinct to $p_i$ and $p_j$}.

\item $f_i$ overwrites $f_j$ (without loss of generality).
Configurations $p_i(C)$ and $p_i p_j(C)$ are indistinguishable to $p_0$ {and all competitors distinct to $p_i$ and $p_j$},
which is a contradiction.
\end{enumerate}

\item $R_i = R_j$ and is a window register:
\begin{enumerate}
\item At least one of $f_i$ or $f_j$ is a read.
Without loss of generality, suppose $f_i$ is a read.
We reach a contradiction because
$p_j p_i (C)$ and $p_j (C)$ are indistinguishable to $p_0$ {and all competitors distinct to $p_i$ and $p_j$}.

\item Both $f_i$ and $f_j$ are writes.
First we make the following observation. Consider any
competitor~$p_k \neq p_i, p_j$, $1 \leq k \leq n$, and let $v_k$ the value $p_k(C)$ is $v_k$-univalent.
Let $R_k$ be the base object~$p_k$ is about to access in $C$, through primitive $f_k$.
Since $v_i \neq v_j$, $v_k$ is distinct to~$v_i$ or~$v_j$.
By cases (1) and (3.a),
if it is not true that $R_k = R_i = R_j$ and~$f_k$ is a write,
then we reach a contradiction (either through $p_i$ or $p_j$).
Thus, in $C$ we have that all processes $p_1, \hdots, p_{n-1}$ are about to write to
the same base object~$R$.
Without loss of generality, let us assume that $p_i = p_1$ and $p_j = p_2$.
Since $R$ is an $(n-2)$-window register,
process $p_0$ {and all competitors distinct to $p_1$ and $p_2$}
cannot distinguish between configurations $p_{n-1} \hdots p_3 p_2 p_1(C)$
and $p_{n-1} \hdots p_3 p_2(C)$, which is a contradiction.
\end{enumerate}

\item $R_i = R_j$ and is a stack.
Without loss of generality,
we assume that in every execution of $A$, every value~$v$
is input of at most one push operation.\footnote{Algorithm $A$ can be modified so that every process $p_i$
stamps input $v$ of a push operation with a pair~$(i, t)$, where $t$ is a local counter incremented after each push operation.
Such a change does not affect $A$'s correctness and progress properties.}
\begin{enumerate}
\item Both $f_i$ and $f_j$ are pop operations. We reach a contradiction in this case because
$p_j p_i (C)$ and $p_i p_j (C)$ are indistinguishable to $p_0$
and all competitors distinct to~$p_i$ and~$p_j$.

\item Both $f_i$ and $f_j$ are push operations. Let $R$ denote $R_i \, (= R_j)$, and $u_i$ and $u_j$ be the input values
of operations $f_i$ and $f_j$, respectively.
Let $p_k$ be any competitor distinct to~$p_i$ and~$p_j$.
Consider any $x \geq 1$ such that $p_k$'s $compete()$ operation is completed in $p_k^x p_j p_i (C)$.
Our analysis depends on whether $p_k$ pops $u_i$ or $u_j$ from $R$ in $p_k^x p_j p_i (C)$:
\begin{enumerate}
\item $p_k$ pops $u_i$ and $u_j$. Thus, there is $1 \leq x' < x$ such that,
 in $p_k^{x'} p_j p_i (C)$, $p_k$~is poised to perform a pop operation in $R$ and $u_i$ is at the top of $R$
 (which implies that $p_k$ already popped~$u_j$ from~$R$).
Due to the specification of~$R$,
neither $p_k$ nor $p_0$ can distinguish between $p_k^{x'+1} p_j p_i (C)$ and $p_k p_i p_k^{x'} p_j (C)$,
which is a contradiction.

\item $p_k$ pops $u_j$ but not $u_i$.
Consider any $y \geq 1$ such that $p_0$'s $decide()$ operation is completed in $p_0^y p_k^x p_j p_i (C)$.
Due to the specification of $R$, if $p_0$ does not pop $u_i$ from $R$ in configuration $p_0^y p_k^x p_j p_i (C)$,
then $p_0$ is in the same local state in that configuration and $p_0^y p_k^x p_j (C)$,
which implies that in both configurations~$p_0$'s operation return the same, but that contradicts the properties of $C$.
Thus, it must be that $p_0$ pops $u_i$ from $R$ in $p_0^y p_k^x p_j p_i (C)$,
and hence there is  $1 \leq y' < y$ such that,
 in $p_0^{y'} p_k^x p_j p_i (C)$, $p_0$ is poised to perform a pop operation in $R$ and $u_i$ is at the top of $R$.
 Again, due to the specification of $R$, $p_0$ cannot distinguish between
 $p_0^{y'+1} p_k^x p_j p_i (C)$ and $p_0 p_i p_0^{y'} p_k^x p_j (C)$, which is a contradiction.

\item $p_k$ does not pop neither $u_i$ nor $u_j$.
Consider any $y \geq 1$ such that $p_0$'s $decide()$ is completed in $p_0^y p_k^x p_j p_i (C)$.
We finalize this case by considering if $p_0$ pops $u_i$ or $u_j$ from $R$ in that configuration.
\begin{itemize}
\item $p_0$ does not pop neither $u_i$ nor $u_j$. Then, $p_0$ is in the same local state in
$p_0^y p_k^x p_j p_i (C)$ and $p_0^y p_k^x p_i p_j (C)$, which derives in a contradiction.

\item $p_0$ pops $u_j$ but not $u_i$. Then, $p_0$ is in the same local state in
$p_0^y p_k^x p_j p_i (C)$ and $p_0^y p_k^x p_j (C)$, which derives in a contradiction too.

\item $p_0$ pops $u_i$ and $u_j$. Here we can reach a contradiction using a similar argument
as in case (4.b.i) above.

\end{itemize}

\end{enumerate}

\item One of $f_i$ and $f_j$ is push and the other pop.
The argument for this case is very similar to the argument for the
previous case.
Without loss of generality, let $f_j$ be push with input $u_j$ and $f_i$ be pop. Let $R$ denote $R_i \, (= R_j)$.
Note that if $R$ is empty in configuration $C$, then neither $p_0$ nor
any competitor distinct to $p_i$ is able to distinguish between $p_j (C)$ and $p_j p_i (C)$,
which derives in a contradiction. Thus, $R$ is not empty in $C$.
Let~$p_k$ be a competitor different from~$p_i$ and~$p_j$.
Let $u_i$ be the value at the top of $R$ in~$p_i (C)$ (possibly the value denoting the stack is empty).
Consider any~$x \geq 1$ such that $p_k$'s $compete()$ operation is completed in $p_k^x p_j p_i (C)$.
We proceed as in case (4.b) above,\footnote{The analysis is exactly the same,
as one can conceive the next step $p_i$ in $C$ is a push operation with input $u_i$.} considering whether $p_k$ pops $u_i$
or $u_j$ from $R$ in $p_k^x p_j p_i (C)$:

\begin{enumerate}
\item $p_k$ pops $u_i$ and $u_j$. Thus, there is $1 \leq x' < x$ such that,
 in $p_k^{x'} p_j p_i (C)$, $p_k$~is poised to perform a pop operation in $R$ and $u_i$ is at the top of $R$
 (which implies that $p_k$ already popped~$u_j$ from~$R$).
Due to the specification of~$R$,
neither $p_k$ nor $p_0$ can distinguish between $p_k^{x'+1} p_j p_i (C)$ and $p_k p_i p_k^{x'} p_j (C)$,
which is a contradiction.

\item $p_k$ pops $u_j$ but not $u_i$.
Consider any $y \geq 1$ such that $p_0$'s $decide()$ operation is completed in $p_0^y p_k^x p_j p_i (C)$.
Due to the specification of $R$, if $p_0$ does not pop $u_i$ from $R$ in configuration $p_0^y p_k^x p_j p_i (C)$,
then $p_0$ is in the same local state in that configuration and $p_0^y p_k^x p_j (C)$,
which implies that in both configurations~$p_0$'s operation return the same, but that contradicts the properties of $C$.
Thus, it must be that $p_0$ pops $u_i$ from $R$ in $p_0^y p_k^x p_j p_i (C)$,
and hence there is  $1 \leq y' < y$ such that,
 in $p_0^{y'} p_k^x p_j p_i (C)$, $p_0$ is poised to perform a pop operation in $R$ and $u_i$ is at the top of $R$.
 Again, due to the specification of $R$, $p_0$ cannot distinguish between
 $p_0^{y'+1} p_k^x p_j p_i (C)$ and $p_0 p_i p_0^{y'} p_k^x p_j (C)$, which is a contradiction.

\item $p_k$ does not pop neither $u_i$ nor $u_j$.
Consider any $y \geq 1$ such that $p_0$'s $decide()$ is completed in $p_0^y p_k^x p_j p_i (C)$.
We finalize this case by considering if $p_0$ pops $u_i$ or $u_j$ from $R$ in that configuration.
\begin{itemize}
\item $p_0$ does not pop neither $u_i$ nor $u_j$. Then, $p_0$ is in the same local state in
$p_0^y p_k^x p_j p_i (C)$ and $p_0^y p_k^x p_i p_j (C)$, which derives in a contradiction.

\item $p_0$ pops $u_j$ but not $u_i$. Then, $p_0$ is in the same local state in
$p_0^y p_k^x p_j p_i (C)$ and $p_0^y p_k^x p_j (C)$, which derives in a contradiction too.

\item $p_0$ pops $u_i$ and $u_j$. Here we can reach a contradiction using a similar argument
as in case (4.b.i) above. \qedhere
\end{itemize}
\end{enumerate}

\end{enumerate}
\end{enumerate}
\end{proof}

\begin{restatable}{proposition}{propqueue}
A lock-free strongly linearizable queue can implement
a wait-free strongly linearizable contest object.
\end{restatable}

\begin{proof}
Let $A$ be any such implementation. Due to strong linearizability,
we can assume that $A$ is an \emph{atomic} base object.
To implement a contest object from $A$,
each competitor $p_i$ simply performs $A.enqueue(i)$,
and the referee $p_0$ performs $A.dequeue()$ and returns the value it obtains,
if distinct from empty, otherwise returns $false$.

As $A$ is lock-free and processes execute at most one operation on $A$,
the implementation is wait-free.
To prove strong linearizability, consider any finite execution $\alpha$ of the implementation.
Let us complete any pending $compete()$ operations that access $A$,
and remove any other pending operation. Let $\alpha$ denote the modified execution.
Every operation of $\alpha$ is linearized at its step that access $A$.
Since $A$ is a linearizable queue,
when $p_0$'s $decide()$ operation returns a value $v \neq false$,
there is a previous $enqueue(v)$ of a competitor, and that must be the first enqueue on $A$.
Thus, the linearization points induce a valid linearization.
The implementation is strongly linearizable because the linearization points in $\alpha$
remain fixed in any extension of it.
\end{proof}

\begin{corollary}
\label{cor:lf sl queue and stack}
There is no lock-free strongly linearizable queue
implementation from interfering primitives and $(n-2)$ window registers and atomic stacks,
for {$n \geq 4$} processes.
\end{corollary}

\section{Impossibility of Wait-Free Strongly Linearizable Implementations: The Long-Lived Contest Object}
\label{sec:wait-free}

The \emph{long-lived contest} object provides two operations, $decide()$ that can be invoked
\emph{at most once} by $p_0$, the \emph{referee}, and $compete()$ that can be invoked \emph{any number of times}
by $p_1, \hdots, p_{n-1}$, the \emph{competitors}.
While $compete()$ returns nothing, $decide()$
returns an integer~$x$ such that every competitor has executed \emph{at most} $x$ operations so far.
Below, the referee is denoted~$p$, and the competitors~$q_1, \hdots, q_{n-1}$.

We first show that there are linearizable long-lived contest implementations from read/write registers,
and strongly linearizable implementations from atomic objects such as counters.

\begin{restatable}{theorem}{llcontestfromreadwrite}
\label{thm:ll-contest-from-read-write}
There is a wait-free long-lived contest (decisively) linearizable implementation
from read/write registers,
for $n \geq 2$ processes.
\end{restatable}

\begin{proof}
Consider the following wait-free, read/write implementation.
The processes share an array $M$ with one entry per competitor (initialized to~0).
Each competitor first increments a local persistent counter (initialized to~0),
and then writes that value to its entry in $M$,
and the referee first collects all values in $M$ and returns the maximum among those values.

We argue that the implementation is decisively linearizable. Consider any finite execution $\alpha$ of the algorithm.
For any pending $compete$ operation in $\alpha$, it is completed if it writes to $M$, otherwise it is removed.
If $\alpha$ has a pending $decide$ operation, this operation is removed. For simplicity, let $\alpha$ itself denote the modified execution.
Operations are linearized as follows: each compete operation is linearized at its write step,
and, if $\alpha$ has a $decide$ operation returning $x$, the operation is linearized at any step in which the maximum value in $M$ is $x$.
To prove linearizability, observe that if a $decide$ operation returns $x$, then the maximum value in $M$ when the operation starts
is at most $x$, and the maximum value in $M$ when the operations completes is at least $x$. Thus, there is a step that lies
between the invocation and response of the operation in which the maximum in $M$ is $x$. This implies that the linearization
points defined above respect the real-time order in $\alpha$, and the induced sequential execution is valid for the long-lived contest object.
Finally, the linearization points also induce a subsequence-closed linearization function: $compete$ operations are appended to
the current linearization, and when the $decide$ operation completes, it is placed somewhere in the current linearization
where its output is justified.
\end{proof}

\begin{restatable}{lemma}{reductions}
\label{lemma:reductions}
For $n \geq 2$ processes,
a wait-free strongly linearizable long-lived contest implementation can be obtained from
a wait-free strongly linearizable counter, snapshot, max register, fetch\&increment or fetch\&add.
\end{restatable}

\begin{proof}
Consider the following wait-free long-lived-contest implementations:
\begin{itemize}
\item Counter. Each competitor executes $increment()$ once, 
and the referee returns the value it obtains in $read()$. 
The implementation is linearizable because, at all times, the counter records the sum of the
number of operations that all competitors have executed so far. Strong linearizability follows from the fact that
the counter is strongly linearizable.

\item Max register. Each competitor first increments a local persistent counter (initialized to~0),
and then writes that value using $maxWrite(\cdot)$, and the referee returns the value it gets from $maxRead()$.
Linearizability follows from the fact that, at all times, the max register object stores the maximum number of operations
that each competitor has executed so far, and strong linearizability follows from strong linearizability of the base object.

\item Snapshot. It is easy to implement a wait-free strongly linearizable counter
or max register from a
single strongly linearizable snapshot.

\item Fetch\&increment. Each competitor performs $fetch\&inc()$, 
and the referee returns the value it obtains in $fetch\&inc()$. 
Linearizability follows from the observation that the object upper bounds
the number operations the competitor have executed. Strong linearizability follows from strong linearizability of the base object.

\item Fetch\&add. Clearly, $fetch\&add(1)$ implements $fetch\&inc()$.\qedhere

\end{itemize}
\end{proof}

We now prove impossibility results for long-lived contest.
Fix a wait-free strongly linearizable long-lived contest implementation $A$ for $n$ processes.
Pick any prefix-closed linearization function $L$ of $A$.
For the time being, we do not care about the primitives
used in $A$ and the value of~$n$.

The proof follows an idea similar to that in~\cite{DenysyukW2015},
and is based on the following notions:

\begin{definition}
Let $\alpha$ be a finite execution.
\begin{itemize}

\item The \emph{competitors-valency} of $\alpha$, denoted $W(\alpha)$,
is the set with every $x$ such that there is a finite competitors-only extension $\beta$
such that $decide()$ appears in $L(\alpha \beta)$ with output $x$.

\item The execution $\alpha$ is \emph{competitors-closed}
if every competitors-only infinite extension $\gamma$ of $\alpha$
has a finite prefix $\beta$ such that $decide()$ appears in $L(\alpha \beta)$.
If $\alpha$ is not competitors-closed, then it is \emph{competitors-supervalent}.

\end{itemize}
\end{definition}

Furthermore, $decide() \in L(\alpha)$ denotes that $decide()$ appears in $L(\alpha)$.
For compactness, prefix \emph{competitors-} is dropped from the notions above.
It directly follows from the definitions:

\begin{proposition}
\label{prop:props-val}
For every execution $\alpha$:

\begin{enumerate}

\item
\label{prop:val-no-lin}
if $W(\alpha) \geq 2$, $decide() \notin L(\alpha)$;

\item
\label{prop:inclusion}
for every competitors-only extension $\beta$ of $\alpha$, $W(\alpha \beta) \subseteq W(\alpha)$;

\item
\label{prop:closed-ext}
if $\alpha$ is closed, for every competitors-only extension~$\beta$,
$\alpha \beta$ is closed;

\item
\label{prop:sup-ext}
if $\alpha$ is supervalent, there is a competitor $q$ such that
$\alpha q$ is supervalent;

\item
\label{prop:sup-non-comp}
if $\alpha$ is supervalent, $decide() \notin L(\alpha)$
(hence $decide()$ has not started or is pending in $\alpha$);

\item
\label{prop:non-start-super}
if $decide()$ is not started in $\alpha$, then $\alpha$ is supervalent;
\end{enumerate}
\end{proposition}

Lemma~\ref{lemma:indist} is used several times, and its proof is a key place
where we use the fact that $A$ is strongly linearizable.
Then, we define the notion of supervalent executions that are helping:
intuitively, no matter what the competitors do, the next step of the referee
decides its current operation.

\begin{restatable}{lemma}{indist}
\label{lemma:indist}
Let $\alpha$ be closed execution that is indistinguishable from an execution $\alpha'$
to the referee~$p$ and a competitor~$q$. Then $W(\alpha) \cap W(\alpha') \neq \emptyset$.
\end{restatable}

\begin{proof}
Let $\gamma$ be the infinite $q$-solo extension of $\alpha$.
Due to wait-freedom, $q$ completes infinitely many operations in $\alpha \gamma$.
Since $\alpha$ is closed, there is a finite prefix $\beta$ of $\gamma$ such that
$decide() \in L(\alpha \beta)$. Let us pick any such $\beta$. As $L$ is prefix-closed,
$decide() \in L(\alpha \beta')$, for every finite prefix $\beta'$ of $\gamma$ having $\beta$ as a prefix.
Let $x$ be the output of $decide()$ in $L(\alpha \beta)$,
and let $\beta'$ be a prefix of $\gamma$ such that $q$ completes more than $x$ operations in
$\alpha \beta'$ and $\beta$ is a prefix of~$\beta'$. Let $\sigma$ be the $p$-solo extension
of $\alpha \beta'$ where $p$'s operation is completed (which must exist due to wait-freedom).
Since $L$ is prefix-closed, the output of $decide()$ in $\alpha \beta' \sigma$ is $x$.
Now observe that $\beta' \sigma$ is an extension of $\alpha'$, as $p$ and $q$ cannot distinguish between
$\alpha$ and $\alpha'$. Thus, in $\alpha' \beta' \sigma$, $q$ completes more than $x$ operations
and $decide()$ is complete with output $x$.
Finally, we note that $decide() \in L(\alpha' \beta')$ and outputs $x$, which implies that $x \in W(\alpha')$,
because if not, then in $L(\alpha' \beta')$ the competitor $q$ completes more than $x$ operations,
and $L(\alpha' \beta' \sigma)$ has $decide()$ with output $x$, which is invalid for the long-lived contest object.
\end{proof}

\begin{definition}
A supervalent execution $\alpha$ is \emph{competitors-helping} (\emph{helping} for short)
if for every competitors-only extension~$\beta$ of $\alpha$,
$\alpha \beta p$ is closed.
\end{definition}

\begin{restatable}{lemma}{helping}
\label{lemma:helping}
There are helping supervalent executions.
\end{restatable}

\begin{proof}
By Proposition~\ref{prop:props-val}(\ref{prop:non-start-super}),
$A$ has supervalent executions.
By contradiction, suppose that none of them is helping.
Then, for every such supervalent execution $\alpha$, there is a competitors-only extension $\beta$ such that
$\alpha \beta p$ is supervalent.
For the empty execution $\alpha_0$, which is supervalent by Proposition~\ref{prop:props-val}(\ref{prop:non-start-super}),
there is a competitors-only extension $\beta_0$
such that $\alpha_1 = \alpha_0 \beta_0 p$ is supervalent.
Note that $p$ takes at least one step in $\alpha_1$,
and by Proposition~\ref{prop:props-val}(\ref{prop:sup-non-comp}), $decide()$ is pending in $\alpha_1$.
For $\alpha_1$, there is a competitors-only extension $\beta_1$
such that $\alpha_2 = \alpha_1 \beta_1 p$ is supervalent.
Note that $p$ takes at least two steps in $\alpha_2$,
and by Proposition~\ref{prop:props-val}(\ref{prop:sup-non-comp}), $decide()$ is pending in $\alpha_2$.
Thus, we can construct an execution $\alpha_\infty$ in which $p$ takes infinitely many steps
and $decide()$ is pending, hence contradicting that $A$ is wait-free.
\end{proof}

\subsection{The case of read/write registers}
\label{sec:proof-read-write}

We now suppose that $A$ uses only read/write registers and $n \geq 3$.
Lemma~\ref{lemma:val-comp-rw} shows that this assumption implies an important property of valencies,
which is then used to characterize valencies of  closed and supervalent executions.

\begin{restatable}{lemma}{valcomprw}
\label{lemma:val-comp-rw}
Let $\alpha$ be an execution such that $\alpha q$ is closed, for a competitor $q$.
For every competitor $q' \neq q$, $W(\alpha q) \cap W(\alpha q') \neq \emptyset$.
\end{restatable}

\begin{proof}
We analyze the steps $q$ and $q'$ are poised to take at the end of $\alpha$,
in all cases concluding that $W(\alpha q) \cap W(\alpha q') \neq \emptyset$:
\begin{itemize}
\item $q$ or $q'$ reads. Suppose first that $q'$ is poised to read.
Then, $\alpha q$ and $\alpha q' q$ are indistinguishable to $p$ and $q$.
By Lemma~\ref{lemma:indist}, $W(\alpha q) \cap W(\alpha q' q) \neq \emptyset$.
By Proposition~\ref{prop:props-val}(\ref{prop:inclusion}),
$W(\alpha q' q) \subseteq W(\alpha q')$, from which follows $W(\alpha q) \cap W(\alpha q') \neq \emptyset$.
Now suppose that $q$ is poised to read.
Then, $\alpha q'$ and $\alpha q q'$ are indistinguishable to $p$ and $q'$.
By Proposition~\ref{prop:props-val}(\ref{prop:closed-ext}), $\alpha q q'$ is closed.
Lemma~\ref{lemma:indist} implies that $W(\alpha q') \cap W(\alpha q q') \neq \emptyset$.
Since $W(\alpha q q') \subseteq W(\alpha q)$,  by Proposition~\ref{prop:props-val}(\ref{prop:inclusion}),
we have $W(\alpha q) \cap W(\alpha q') \neq \emptyset$.

\item $q$ and $q'$ write to distinct registers.
Thus, $\alpha q q'$ and $\alpha q' q$ are indistinguishable to~$p$ and~$q$ (and~$q'$).
By Lemma~\ref{lemma:indist}, $W(\alpha q q') \cap W(\alpha q' q) \neq \emptyset$.
By Proposition~\ref{prop:props-val}(\ref{prop:closed-ext}), $\alpha q q'$ is closed.
Since $W(\alpha q) \subseteq W(\alpha q q')$ and $W(\alpha q') \subseteq W(\alpha q' q)$,
by Proposition~\ref{prop:props-val}(\ref{prop:inclusion}),
we have $W(\alpha q) \cap W(\alpha q') \neq \emptyset$.

\item $q$ and $q'$ write to the same register.
Then, $\alpha q$ and $\alpha q' q$ are indistinguishable to $p$ and $q$.
We have that $W(\alpha q) \cap W(\alpha q' q) \neq \emptyset$,
by Lemma~\ref{lemma:indist}.
By Proposition~\ref{prop:props-val}(\ref{prop:inclusion}),
$W(\alpha q' q) \subseteq W(\alpha q')$, hence $W(\alpha q) \cap W(\alpha q') \neq \emptyset$.\qedhere
\end{itemize}
\end{proof}

\begin{restatable}{lemma}{lemclosed}
\label{lemma:closed}
If $\alpha$ is closed, $|W(\alpha)| = 1$.
\end{restatable}

\begin{proof}
By contradiction, suppose that $|W(\alpha)| \geq 2$. We extend $\alpha$ as follows:
we let a competitor take its next step as long as this yields an extension $\beta$
with $|W(\alpha \beta)| \geq 2$. Observe that such extension~$\beta$ cannot be infinite,
because if so, for every finite prefix $\beta'$ of $\beta$ we have that $decide() \notin L(\alpha \beta')$,
by Proposition~\ref{prop:props-val}(\ref{prop:val-no-lin}), hence contradicting that $\alpha$ is closed.
Pick any such finite extension $\beta$. Thus, $|W(\alpha \beta)| \geq 2$, and for every competitor $q$,
$|W(\alpha \beta q)| = 1$. By Proposition~\ref{prop:props-val}(\ref{prop:closed-ext}), $\alpha \beta q$ is closed,
for every competitor $q$.
Proposition~\ref{prop:props-val}(\ref{prop:inclusion}) implies that there
are distinct competitors $q$ and $q'$ such that $W(\alpha \beta q) \cap W(\alpha \beta q') = \emptyset$.
This contradicts Lemma~\ref{lemma:val-comp-rw}.
\end{proof}

\begin{restatable}{lemma}{supervalent}
\label{lemma:supervalent}
If $\alpha$ is supervalent, $|W(\alpha)| = 0$.
\end{restatable}

\begin{proof}
Suppose  otherwise, and consider any $x \in W(\alpha)$.
We extend $\alpha$ as follows:
we let a competitor take its next step as long as this yields an extension $\beta$
such that $\alpha \beta$ is supervalent and $x \in W(\alpha \beta)$.
Such extension $\beta$ cannot be infinite,
because if so, there is a finite prefix $\beta'$ of $\beta$ such that more than
$x$ operations of a single competitor are completed in $\alpha \beta'$, due to wait-freedom,
and $decide() \notin L(\alpha \beta')$, by Proposition~\ref{prop:props-val}(\ref{prop:sup-non-comp}).
As $L$ is a linearization function, there is no extension of $\alpha \beta'$ in which
$decide()$ returns $x$, contradicting that $x \in W(\alpha \beta')$.
Pick any such finite extension $\beta$. Thus, $x \in W(\alpha \beta)$, and for any competitor $q$,
$\alpha \beta q$ is closed or $x \notin W(\alpha \beta q)$.
By Proposition~\ref{prop:props-val}(\ref{prop:sup-ext}), there is a competitor~$q'$
such that~$\alpha \beta q'$ is supervalent. Hence, $x \notin W(\alpha \beta q')$.
Proposition~\ref{prop:props-val}(\ref{prop:inclusion}) implies that there is a competitor $q \neq q'$
such that $x \in W(\alpha \beta q)$. Then, $\alpha \beta q$ is closed.
Lemma~\ref{lemma:closed}
implies that $W(\alpha \beta q') = \{x\}$,
and hence, $W(\alpha \beta q) \cap W(\alpha \beta q') = \emptyset$.
This contradicts Lemma~\ref{lemma:val-comp-rw}.
\end{proof}

\begin{restatable}{lemma}{valrefrw}
\label{lemma:val-ref-rw}
Let $\alpha$ be an execution such that $\alpha p$ is closed and,
for a competitor $q$, it is \emph{not} the case that at the end of $\alpha$,
the referee $p$ is poised to read a register and $q$ is poised to write to the same register.
Then, $W(\alpha p) \cap W(\alpha q p) \neq \emptyset$.
\end{restatable}

\begin{proof}
We analyze the steps $p$ and $q$ are poised to take at the end of $\alpha$,
in all cases concluding that $W(\alpha p) \cap W(\alpha q p) \neq \emptyset$:
\begin{itemize}
\item $q$ reads. Consider a competitor $q' \neq q$.
Then, $p$ and $q'$ cannot distinguish between $\alpha p$ and $\alpha q p$.
By Lemma~\ref{lemma:indist}, $W(\alpha p) \cap W(\alpha q p) \neq \emptyset$.

\item $q$ and $p$ apply primitives to distinct registers.
Then, $p$ and $q$ cannot distinguish between $\alpha p q$ and $\alpha q p$.
Proposition~\ref{prop:props-val}(\ref{prop:closed-ext}) implies that $\alpha p q$ is closed.
By Lemma~\ref{lemma:indist}, $W(\alpha p q) \cap W(\alpha q p) \neq \emptyset$.
Since $W(\alpha p q) \subseteq W(\alpha p)$, by Proposition~\ref{prop:props-val}(\ref{prop:inclusion}),
we have $W(\alpha p) \cap W(\alpha q p)~\neq~\emptyset$.

\item  $p$ and $q$ write to the same register.
For every competitor $q' \neq q$,
$p$ and $q'$ cannot distinguish between $\alpha p$ and $\alpha q p$.
By Lemma~\ref{lemma:indist}, $W(\alpha p) \cap W(\alpha q p) \neq \emptyset$.\qedhere

\end{itemize}

\end{proof}

By Lemma~\ref{lemma:helping}, there is a helping supervalent execution $\alpha$,
namely, for every competitors-only extension~$\beta$, $\alpha \beta p$ is closed.
For the rest of the proof, fix any such execution $\alpha$.

\begin{lemma}
\label{lemma:to-reach-contr}
There is a finite competitors-only extension $\beta$ of $\alpha$ such that, for every competitor~$q$,
\begin{enumerate}
\item
\label{prop:contr-simple}
$W(\alpha \beta p) \cap W(\alpha \beta q p) = \emptyset$, and

\item
\label{prop:contr-complex}
for every competitor $q' \neq q$, there is a $k \geq 0$ such that
$W(\alpha \beta q p q'^{k+1}) \cap W(\alpha \beta q'^k q p q') = \emptyset$ and
$W(\alpha \beta q'^{k+1} q p ) \cap W(\alpha \beta q'^k q p q') = \emptyset$.
\end{enumerate}
\end{lemma}

\begin{proof}
Let $\beta$ be any competitors-only extension of $\alpha$.
By Lemma~\ref{lemma:closed}, $|W(\alpha \beta p)| = 1$.
By Lemma~\ref{lemma:supervalent}, $W(\alpha)~=~\emptyset$, and
hence, by the definition of valency,
$decide() \notin L(\alpha \beta)$. Let $W(\alpha p) = \{x\}$.
We extend $\alpha$ as follows:
we let a competitor take its next step as long as this yields a competitors-only extension $\beta$
such that $W(\alpha \beta p) = \{x\}$.
Such extension $\beta$ cannot be infinite,
because if so, there is a finite prefix $\beta'$ of $\beta$ such that
$W(\alpha \beta' p) = \{x\}$, $decide() \notin L(\alpha \beta')$
and a competitor completes more than $x$ operations in $\alpha \beta'$;
and as $L$ is a linearization function, there is no extension of $\alpha \beta' p$ in which
$decide()$ returns $x$, contradicting that $W(\alpha \beta' p) = \{x\}$.
Pick any such finite extension $\beta$.
Thus, we have $W(\alpha \beta p) = \{x\}$,
and for every competitor $q$, $W(\alpha \beta q p) \neq \{x\}$.
By Lemma~\ref{lemma:closed}, $|W(\alpha \beta q p)| = 1$, hence $W(\alpha \beta p) \cap W(\alpha \beta q p) = \emptyset$, showing (1).

To prove (2), let $W(\alpha \beta q p) = \{y\}$.
Let $\ell \geq 0$. Proposition~\ref{prop:props-val}(\ref{prop:closed-ext}) implies that
$\alpha \beta q p q'^{\ell+1}$ is closed,
and Proposition~\ref{prop:props-val}(\ref{prop:inclusion}) and Lemma~\ref{lemma:closed}
that $W(\alpha \beta q p q'^{\ell+1}) = \{y\}$.
By Lemma~\ref{lemma:closed}, $|W(\alpha \beta q'^\ell q p)|~=~1$. Let $W(\alpha \beta q'^\ell q p) = \{z\}$.
It was observed too that $decide() \notin L(\alpha \beta q'^\ell q)$.
Due to wait-freedom, there are values of $\ell$ such that $q'$ completes more than $y$ operations in $\alpha \beta q'^\ell q$.
Pick any such $\ell$, and consider any $k \geq \ell$.
Due to the specification of long-lived contest and
since $L$ is a linearization function, there is no extension of $\alpha \beta q'^k q$ in which $decide()$ outputs $y$.
Thus, $y \neq z$.
By Proposition~\ref{prop:props-val}(\ref{prop:inclusion}),
$W(\alpha \beta q'^k q p q') = W(\alpha \beta q'^k q p)$,
and hence
$$W(\alpha \beta q p q'^{k+1}) \cap W(\alpha \beta q'^k q p q') = \emptyset,$$
for any $k \geq \ell$.
Using similar reasoning, we can argue that it is not true that
$$W(\alpha \beta q'^k q p) = W(\alpha \beta q'^{k+1} q p),$$
for every $k \geq \ell$.
Thus, there is a $k \geq \ell$ such that $W(\alpha \beta q'^k q p) \neq W(\alpha \beta q'^{k+1} q p)$,
and moreover $W(\alpha \beta q'^k q p) \cap W(\alpha \beta q'^{k+1} q p) = \emptyset$.
Finally, Proposition~\ref{prop:props-val}(\ref{prop:inclusion}) and Lemma~\ref{lemma:closed} imply that
$W(\alpha \beta q'^k q p q') = W(\alpha \beta q'^k q p)$.
\end{proof}

To reach a contradiction, consider a competitors-only 
extension $\beta$ of $\alpha$ as stated in Lemma~\ref{lemma:to-reach-contr}.
We have that $\alpha \beta p$ is closed.
Otherwise, if at the end of $\alpha \beta$
all competitors are poised to write to a register $R$ and the referee is poised to read~$R$,
then by Lemma~\ref{lemma:val-ref-rw}, 
$W(\alpha \beta q p q'^{k+1}) \cap W(\alpha \beta q'^k q p q') = \emptyset$,
contradicting Lemma~\ref{lemma:to-reach-contr}(\ref{prop:contr-simple}).
Thus, it must be that
all competitors are poised to write to a register $R$ and the referee is poised to read~$R$, at the end of $\alpha \beta$.
We use the properties of $\beta$ stated in 
Lemma~\ref{lemma:to-reach-contr}(\ref{prop:contr-complex}) 
to reach a contradiction.

First, we use $W(\alpha \beta q p q'^{k+1}) \cap W(\alpha \beta q'^k q p q') = \emptyset$.
Observe that, at the end of $\alpha \beta q p q'^{k+1}$ and $\alpha \beta q'^k q p q'$,
$q$ is in the same state as it only writes to $R$ after $\alpha \beta$,
and $p$ is in the same state too as it reads from $R$ what $q$ writes.
As after $\alpha \beta$, $q'$ is poised to write to $R$,
its state and the state of the shared memory is the same at the end of $\alpha \beta q p q'^k$ and $\alpha \beta q'^k$.
For the time being, suppose that $q'$ is poised to write to $R$ at the end of $\alpha \beta q'^k$;
this implies that  $q'$ is poised to write to $R$ at the end of $\alpha \beta q p q'^k$ too.
Then, the state of the shared memory is the same at the end of the two executions (and the state of $q'$ too),
and thus $p$ and $q$ cannot distinguish between the executions.
Since $\alpha \beta q p q'^{k+1}$ is closed,
this contradicts Lemma~\ref{lemma:indist}.

We complete the proof
by arguing that it is actually the case that $q'$ is poised
to write to $R$ at the end of $\alpha \beta q'^k$.
We now use that $W(\alpha \beta q'^{k+1} q p ) \cap W(\alpha \beta q'^k q p q') = \emptyset$.
Using a similar reasoning as in the paragraph above, we can see that
if $q'$ is poised to read $R$ or access a register $R' \neq R$, then
$p$ and $q$ cannot distinguish between $\alpha \beta q'^{k+1} q p$ and  $\alpha \beta q'^k q p q'$,
which contradicts Lemma~\ref{lemma:indist}.
Therefore we have:

\begin{theorem}
\label{thm:no ll contest from read write}
There is no wait-free strongly linearizable long-lived contest implementation
from only read and write primitives, for $n \geq 3$ processes.
\end{theorem}

Lemma~\ref{lemma:reductions} and the previous theorem imply:

\begin{corollary}
\label{cor:no wf sl counter from rw}
If $n \geq 3$, there is no wait-free strongly linearizable implementation of snapshot, max register, counter,
fetch\&increment or fetch\&add that uses only read/write registers.
\end{corollary}

\subsection{The case of window registers and test\&set}

Suppose that $A$ uses $w$-window registers and test\&set objects.
The next extends the arguments in Section~\ref{sec:proof-read-write},
with a slightly different structure.
Lemmas~\ref{lemma:closed-wind-ts} and~\ref{lemma:supervalent-wind-ts}
show that competitors valency has properties similar to
those in Lemmas~\ref{lemma:closed} and~\ref{lemma:supervalent} above.

\begin{restatable}{lemma}{closedwindts}
\label{lemma:closed-wind-ts}
Suppose that $n \geq 4$ and $w \leq n-2$.
If $\alpha$ is closed, then $|W(\alpha)| = 1$.
\end{restatable}

\begin{proof}
By contradiction, suppose that $|W(\alpha)| \geq 2$. We extend $\alpha$ as follows:
we let a competitor take its next step as long as this yields an extension $\beta$
with $|W(\alpha \beta)| \geq 2$. Observe that such extension $\beta$ cannot be infinite,
because if so, for every finite prefix $\beta'$ of $\beta$ we have that $decide() \notin L(\alpha \beta')$,
by Proposition~\ref{prop:props-val}(\ref{prop:val-no-lin}), hence contradicting that $\alpha$ is closed.
Pick any such finite extension $\beta$. Thus, $|W(\alpha \beta)| \geq 2$, and for every competitor $q$,
$|W(\alpha \beta q)| = 1$. By Proposition~\ref{prop:props-val}(\ref{prop:closed-ext}), $\alpha \beta$ is closed.
Proposition~\ref{prop:props-val}(\ref{prop:inclusion}) implies that there
are competitors $q$ and $q'$ such that $W(\alpha \beta q) \cap W(\alpha \beta q') = \emptyset$.
We complete the proof by analyzing the steps $q$ and $q'$ are poised to take at the end of
$\alpha \beta$.

\begin{itemize}

\item $q$ or $q'$ reads. Suppose, without loss of generality, that $q'$ is poised to read.
Then, $\alpha \beta q$ and $\alpha \beta q' q$ are indistinguishable to $p$ and $q$.
We have that $\alpha \beta q$ is closed, by Proposition~\ref{prop:props-val}(\ref{prop:closed-ext}),
and $W(\alpha \beta q) \cap W(\alpha \beta q' q)~=~\emptyset$, by Proposition~\ref{prop:props-val}(\ref{prop:inclusion}).
This contradicts Lemma~\ref{lemma:indist}.

\item $q$ and $q'$ access distinct base objects.
Thus, $\alpha \beta q q'$ and $\alpha \beta q' q$ are indistinguishable to $p$ and $q$ (and $q'$).
Proposition~\ref{prop:props-val}(\ref{prop:closed-ext}) implies that $\alpha \beta q q'$ is closed,
and Proposition~\ref{prop:props-val}(\ref{prop:inclusion}) gives that  $W(\alpha \beta q q') \cap W(\alpha \beta q' q)~=~\emptyset$.
This contradicts Lemma~\ref{lemma:indist}.

\item $q$ and $q'$ $T\&S()$ to the same test\&set object.
Consider any competitor $q'' \neq q,q'$ (which exists as $n \geq 4$).
Then, $\alpha \beta q$ and $\alpha \beta q'$ are indistinguishable to $p$ and $q''$.
Proposition~\ref{prop:props-val}(\ref{prop:closed-ext}) implies that $\alpha \beta q$ and $\alpha \beta q'$ are closed.
As already said, $W(\alpha \beta q) \cap W(\alpha \beta q') = \emptyset$, which contradicts Lemma~\ref{lemma:indist}.

\item $q$ and $q'$ write to the same register $R$.
Let $q''$ be any competitor distinct to $q$ and $q'$.
Since $|W(\alpha \beta q'')| = 1$,
it must be that $W(\alpha \beta \bar q) \cap W(\alpha \beta q'') = \emptyset$,
for $\bar q$ equal to $q$ or $q'$.
Thus, the analysis so far for $\bar q$ and $q''$ shows that, at the end of $\alpha \beta$,
$q''$ must be poised to write to~$R$ (otherwise we reach a contradiction).
Then, all competitors are poised to write to $R$, at the end of $\alpha \beta$.
Since $w \leq n-2 $, for any sequence $\gamma$ with all the $n-3$ competitors distinct to $q$ and $q'$,
$\alpha \beta q \gamma$ and $\alpha \beta q' q \gamma$ are indistinguishable to $p$ and $q$.
Proposition~\ref{prop:props-val}(\ref{prop:inclusion}) and Proposition~\ref{prop:props-val}(\ref{prop:closed-ext})
imply that both executions $\alpha \beta q \gamma$ and $\alpha \beta q' q \gamma$ are closed with
$W(\alpha \beta q \gamma) \cap W(\alpha \beta q' q \gamma)~=~\emptyset$.
This contradicts Lemma~\ref{lemma:indist}.\qedhere
\end{itemize}
\end{proof}

\begin{restatable}{lemma}{supervalentwindts}
\label{lemma:supervalent-wind-ts}
Suppose that $n \geq 4$ and $w \leq n-2$.
If $\alpha$ is supervalent, then $|W(\alpha)| = 0$.
\end{restatable}

\begin{proof}
Suppose  otherwise, and consider any $x \in W(\alpha)$.
We extend $\alpha$ as follows:
we let a competitor take its next step as long as this yields an extension $\beta$
such that $\alpha \beta$ is supervalent and $x \in W(\alpha \beta)$.
Such extension $\beta$ cannot be infinite,
because if so, there is a finite prefix $\beta'$ of $\beta$ such that more than
$x$ operations of a single competitor are completed in $\alpha \beta'$, due to wait-freedom,
and $decide() \notin L(\alpha \beta')$, by Proposition~\ref{prop:props-val}(\ref{prop:sup-non-comp});
and as $L$ is a linearization function, there is no extension of $\alpha \beta'$ in which
$decide()$ returns $x$, contradicting that $x \in W(\alpha \beta')$.
Pick any such finite extension $\beta$. Thus, $x \in W(\alpha \beta)$, and for any competitor $q$,
$\alpha \beta q$ is closed or $x \notin W(\alpha \beta q)$.
By Proposition~\ref{prop:props-val}(\ref{prop:sup-ext}), there is a competitor $q'$
such that $\alpha \beta q'$ is supervalent. Hence, $x \notin W(\alpha \beta q')$.
Proposition~\ref{prop:props-val}(\ref{prop:inclusion}) implies that there is a competitor $q \neq q'$
such that $x \in W(\alpha \beta q)$. Then, $\alpha \beta q$ is closed.
Lemma~\ref{lemma:closed-wind-ts}
gives that actually $W(\alpha \beta q) = \{x\}$,
which implies that $W(\alpha \beta q) \cap W(\alpha \beta q') = \emptyset$.
We complete the proof by analyzing the steps $q$ and $q'$ are poised to take at the end of
$\alpha \beta$.

\begin{itemize}

\item $q$ or $q'$ reads. Suppose first that $q'$ is poised to read.
Then, $\alpha \beta q$ and $\alpha \beta q' q$ are indistinguishable to $p$ and $q$.
We have that $\alpha \beta q$ is closed, as already claimed,
and $W(\alpha \beta q) \cap W(\alpha \beta q' q)~=~\emptyset$, by Proposition~\ref{prop:props-val}(\ref{prop:inclusion}).
This contradicts Lemma~\ref{lemma:indist}.
Now suppose that $q$ is poised to read.
Then, $\alpha \beta q'$ and $\alpha \beta q q'$ are indistinguishable to $p$ and $q'$.
By Proposition~\ref{prop:props-val}(\ref{prop:closed-ext}), $\alpha \beta q q'$ is closed,
and by Proposition~\ref{prop:props-val}(\ref{prop:inclusion}), $W(\alpha \beta q q') \cap W(\alpha \beta q')~=~\emptyset$.
This contradicts Lemma~\ref{lemma:indist}.

\item $q$ and $q'$ access distinct base objects.
Thus, $\alpha \beta q q'$ and $\alpha \beta q' q$ are indistinguishable to $p$ and $q$ (and $q'$).
Proposition~\ref{prop:props-val}(\ref{prop:closed-ext}) implies that $\alpha \beta q q'$ is closed,
and Proposition~\ref{prop:props-val}(\ref{prop:inclusion}) gives that  $W(\alpha \beta q q') \cap W(\alpha \beta q' q)~=~\emptyset$.
This contradicts Lemma~\ref{lemma:indist}.

\item $q$ and $q'$ $T\&S()$ to the same test\&set object.
Consider any competitor $q'' \neq q,q'$ (which exists as $n \geq 4$).
Then, $\alpha \beta q$ and $\alpha \beta q'$ are indistinguishable to $p$ and $q''$.
Proposition~\ref{prop:props-val}(\ref{prop:closed-ext}) implies that $\alpha \beta q$ and $\alpha \beta q'$ are closed.
As already said, $W(\alpha \beta q) \cap W(\alpha \beta q') = \emptyset$, which contradicts Lemma~\ref{lemma:indist}.

\item $q$ and $q'$ write to the same register $R$.
Let $q''$ be any competitor distinct to $q$ and $q'$.
If $x \notin W(\alpha \beta q'')$, then the analysis above for $q$ and $q''$ shows that
$q''$ is poised to write to $R$ (otherwise we reach a contradiction).
If $x \in W(\alpha \beta q'')$, then we have have that $\alpha \beta q''$ is closed,
hence $W(\alpha \beta q'') = \{x\}$, by Lemma~\ref{lemma:closed-wind-ts},
from which follows that $W(\alpha \beta q'') \cap W(\alpha \beta q') = \emptyset$.
Again, the analysis above for $q''$ and $q'$ gives that $q''$ is poised to write to $R$.
Thus, all competitors are poised to write to $R$, at the end of $\alpha \beta$.
Since $w \leq n-2$, for any sequence $\gamma$ with all the $n-3$ competitors distinct to $q$ and $q'$,
$\alpha \beta q \gamma$ and $\alpha \beta q' q \gamma$ are indistinguishable to $p$ and $q$.
We have that $\alpha \beta q$ is closed, as already mentioned, hence $\alpha \beta q \gamma$ is closed, by Proposition~\ref{prop:props-val}(\ref{prop:closed-ext}),
and furthermore $W(\alpha \beta q \gamma) \cap W(\alpha \beta q' q \gamma)~=~\emptyset$, by Proposition~\ref{prop:props-val}(\ref{prop:inclusion}).
This contradicts Lemma~\ref{lemma:indist}.\qedhere
\end{itemize}
\end{proof}

By Lemma~\ref{lemma:helping}, there is a helping supervalent execution $\alpha$,
namely, for every competitors-only extension~$\beta$, $\alpha \beta p$ is closed.
For the rest of the proof, fix such an execution $\alpha$.
We remark that this is the only place where the restriction on the size of
the window register, $w$, is required.

\begin{restatable}{lemma}{contrwindts}
\label{lemma:to-reach-contr-wind-ts}
Suppose that $n \geq 4$ and $3w \leq n$.
There is a finite competitors-only extension $\beta$ such that, for every competitor $q$,
\begin{enumerate}
\item
\label{prop:contr-simple-wind-ts}
$W(\alpha \beta p) \cap W(\alpha \beta q p) = \emptyset$, and

\item
\label{prop:contr-complex-wind-ts}
for every competitor $q' \neq q$, there is a $k \geq 0$,
and sequences $\gamma, \gamma'_1, \gamma'_2$ of competitors
with $|\gamma| = |\gamma'_1| = |\gamma'_2| = w$, $\gamma = q \lambda$, $\gamma'_1 = q' \lambda'_1$, $\gamma'_2 = q' \lambda'_2$, and
$\lambda, \lambda'_1, \lambda'_2$ being competitor-disjoint,
such that
$W(\alpha \beta \gamma p \gamma'_1 q'^k \gamma'_2) \cap W(\alpha \beta \gamma'_1 q'^k \gamma p \gamma'_2) = \emptyset$ and
$W(\alpha \beta \gamma'_1 q'^{k+1} \gamma p ) \cap W(\alpha \beta \gamma'_1 q'^k \gamma p q') = \emptyset$.
\end{enumerate}
\end{restatable}

\begin{proof}
Let $\beta$ be any competitors-only extension of $\alpha$.
By Lemma~\ref{lemma:closed-wind-ts}, $|W(\alpha \beta p)| = 1$.
By Lemma~\ref{lemma:supervalent-wind-ts}, $W(\alpha)~=~\emptyset$, and
hence, by the definition of valency,
$decide() \notin L(\alpha \beta)$. Let $W(\alpha p) = \{x\}$.
We extend $\alpha$ as follows:
we let a competitor take its next step as long as this yields a competitors-only extension $\beta$
such that $W(\alpha \beta p) = \{x\}$.
Such extension $\beta$ cannot be infinite,
because if so, there is a finite prefix $\beta'$ of $\beta$ such that
$W(\alpha \beta' p) = \{x\}$, $decide() \notin L(\alpha \beta')$
and a competitor completes more than $x$ operations in $\alpha \beta'$;
and as $L$ is a linearization function, there is no extension of $\alpha \beta' p$ in which
$decide()$ returns $x$, contradicting that $W(\alpha \beta' p) = \{x\}$.
Pick any such finite extension $\beta$.
Thus, we have $W(\alpha \beta p) = \{x\}$,
and for every competitor $q$, $W(\alpha \beta q p) \neq \{x\}$.
As already said, $|W(\alpha \beta q p)| = 1$, hence $W(\alpha \beta p) \cap W(\alpha \beta q p) = \emptyset$.

To prove the second claim, first note the assumption $3w \leq n$ implies that there are
sequences $\gamma, \gamma'_1, \gamma'_2$ of competitors satisfying that
$|\gamma| = |\gamma'_1| = |\gamma'_2| = w$,
$\gamma = q \lambda$, $\gamma'_1 = q' \lambda'_1$, $\gamma'_2 = q' \lambda'_2$, and
$\lambda, \lambda'_1, \lambda'_2$ being competitor-disjoint. Pick any such sequences.
Let $W(\alpha \beta \gamma p) = \{y\}$.
Let $\ell \geq 0$. Proposition~\ref{prop:props-val}(\ref{prop:closed-ext}) implies that
$\alpha \beta \gamma p \gamma'_1 q'^\ell \gamma'_2$ is closed,
and Proposition~\ref{prop:props-val}(\ref{prop:inclusion}) and Lemma~\ref{lemma:closed-wind-ts}
that $W(\alpha \beta \gamma p \gamma'_1 q'^\ell \gamma'_2) = \{y\}$.
As already observed, $|W(\alpha \beta \gamma'_1 q'^\ell \gamma p \gamma'_2)|~=~1$.
Let $W(\alpha \beta \gamma'_1 q'^\ell \gamma p \gamma'_2) = \{z\}$.
It was observed too that $decide() \notin L(\alpha \beta \gamma'_1 q'^\ell \gamma)$.
Due to wait-freedom, there are values of $\ell$ such that $q'$ completes more than $y$ operations in $\alpha \beta \gamma'_1 q'^\ell \gamma$.
Pick any such $\ell$, and consider any $k \geq \ell$.
Since $L$ is a linearization function, there is no extension of $\alpha \beta \gamma'_1 q'^k \gamma$ in which $decide()$ outputs $y$.
Thus, $y \neq z$,
and hence $W(\alpha \beta \gamma p \gamma'_1 q'^k \gamma'_2) \cap W(\alpha \beta \gamma'_1 q'^k \gamma p \gamma'_2) = \emptyset$,
for any $k \geq \ell$.
Using similar reasoning, we can argue that it is not true
$W(\alpha \beta \gamma'_1 q'^k \gamma p) = W(\alpha \beta \gamma'_1 q'^{k+1} \gamma p)$, for every $k \geq \ell$.
Thus, there is a $k \geq \ell$ such that $W(\alpha \beta \gamma'_1 q'^k \gamma p) \neq W(\alpha \beta \gamma'_1 q'^{k+1} \gamma p)$,
and moreover $W(\alpha \beta \gamma'_1 q'^k \gamma p) \cap W(\alpha \beta \gamma'_1 q'^{k+1} \gamma p) = \emptyset$.
Finally, by Proposition~\ref{prop:props-val}(\ref{prop:inclusion}) and Lemma~\ref{lemma:closed},
$W(\alpha \beta \gamma'_1 q'^k \gamma p q') = W(\alpha \beta \gamma'_1 q'^k \gamma p)$.
\end{proof}

To reach a contradiction, assume that $n \geq 4$ and $3w \leq n$,
and consider a competitors-only extension $\beta$ of $\alpha$ as stated in Lemma~\ref{lemma:to-reach-contr-wind-ts}.
We complete the proof by analyzing the steps processes are poised to take at the end of
$\alpha \beta$, in each case reaching a contradiction using Lemma~\ref{lemma:indist}.
In the first three cases we use the property of $\beta$ stated in Lemma~\ref{lemma:to-reach-contr-wind-ts}(\ref{prop:contr-simple-wind-ts}),
and in last two cases we additionally use the property stated in Lemma~\ref{lemma:to-reach-contr-wind-ts}(\ref{prop:contr-complex-wind-ts}).

\begin{itemize}

\item A competitor is poised to read. Let $q$ and $q'$ be competitors
such that $q$ is poised to read and $q' \neq q$.
Then, $p$ and $q'$ cannot distinguish between $\alpha \beta p$ and $\alpha \beta q p$,
which contradicts Lemma~\ref{lemma:indist},
as $\alpha \beta p$ is closed and $W(\alpha \beta p) \cap W(\alpha \beta q p) = \emptyset$.

\item A competitor $q$ and the referee $p$ access distinct base objects.
No process can distinguish between $\alpha \beta p q$ and $\alpha \beta q p$.
Proposition~\ref{prop:props-val}(\ref{prop:closed-ext}) implies that $\alpha \beta p q$ is closed,
and Proposition~\ref{prop:props-val}(\ref{prop:inclusion}) gives that
$W(\alpha \beta p q) \cap W(\alpha \beta q p)~=~\emptyset$.
This contradicts Lemma~\ref{lemma:indist}.

\item  All processes write to the same register.
For every pair of distinct competitors $q$ and $q'$,
$p$ and $q'$ cannot distinguish between $\alpha \beta p \gamma$ and $\alpha \beta q p \gamma$,
where $\gamma$ is any sequence with all the $n-3$ competitors distinct to $q$ and $q'$.
Propositions~\ref{prop:props-val}(\ref{prop:inclusion}-\ref{prop:closed-ext})
imply that  $\alpha \beta p \gamma$ and $\alpha \beta q p \gamma$ are closed and
$W(\alpha \beta p \gamma) \cap W (\alpha \beta q p \gamma) = \emptyset$,
which contradicts Lemma~\ref{lemma:indist}.

\item  All processes $T\&S()$ the same test\&set base object $R$.
By Lemma~\ref{lemma:to-reach-contr-wind-ts}(\ref{prop:contr-complex-wind-ts}), we have
$W(\alpha \beta \gamma p \gamma'_1 q'^k \gamma'_2) \cap W(\alpha \beta \gamma'_1 q'^k \gamma p \gamma'_2) = \emptyset$.
At the end executions $\alpha \beta \gamma p \gamma'_1 q'^k \gamma'_2$ and
$\alpha \beta \gamma'_1 q'^k \gamma p \gamma'_2$,
$p$ is in the same state, as in both of them it obtains $false$ from its $R.T\&S()$.
Note that the state of the shared memory is the same at the end of both executions.
Moreover, the executions are indistinguishable to any
competitor $q'' \neq q,q'$, which exists since $n \geq 4$.
Since the two executions are closed, this contradicts Lemma~\ref{lemma:indist}.

\item  All competitors write to $R$, and the referee reads $R$.
Consider again executions $\alpha \beta \gamma p \gamma'_1 q'^k \gamma'_2$ and $\alpha \beta \gamma'_1 q'^k \gamma p \gamma'_2$.
Observe that, at the end of both executions,
$q$ is in the same state as it only writes to $R$ after $\alpha \beta$,
and $p$ is in the same state too as it reads from $R$ what the processes in $\gamma$ write.
As after $\alpha \beta$, the competitors in $\gamma'_1$ are poised to write to $R$,
their state and the state of the shared memory is the same at the end of
$\alpha \beta \gamma p \gamma'_1 q'^k$ and $\alpha \beta \gamma'_1 q'^k$.
For the time being, suppose that $q'$ is poised to write to $R$ at the end of $\alpha \beta \gamma'_1 q'^k$;
this implies that  $q'$ is poised to write to $R$ at the end of $\alpha \beta \gamma p \gamma'_1 q'^k$ too.
Then, the state of the shared memory is the same at the end of
$\alpha \beta \gamma p \gamma'_1 q'^k \gamma'_2$ and $\alpha \beta \gamma'_1 q'^k \gamma p \gamma'_2$ (and the state of $q'$ too),
and thus $p$ and $q$ cannot distinguish between the executions.
Since both executions are closed, which contradicts Lemma~\ref{lemma:indist}.

We complete the proof of this case by arguing that indeed it is the case that $q'$ is poised to write to $R$ at the end of $\alpha \beta \gamma'_1 q'^k$.
We now use that $W(\alpha \beta \gamma'_1 q'^{k+1} \gamma p) \cap W(\alpha \beta \gamma'_1 q'^k \gamma p q') = \emptyset$,
by Lemma~\ref{lemma:to-reach-contr-wind-ts}(\ref{prop:contr-complex-wind-ts}).
Using a similar line of reasoning as in the paragraph above, we can see that
if $q'$ is poised to read $R$ or access an object $R' \neq R$, then
$p$ and $q$ cannot distinguish between $\alpha \beta \gamma'_1 q'^{k+1} \gamma p$ and  $\alpha \beta \gamma'_1 q'^k \gamma p q'$,
which contradicts Lemma~\ref{lemma:indist}.
\end{itemize}

In all cases we reach a contradiction. Thus, we have:

\begin{theorem}
\label{thm:ll-contest-from-window-ts}
If $n \geq 4$ and $3w \leq n$, the implementation $A$ cannot use only $w$-window registers and test\&set objects.
\end{theorem}

Lemma~\ref{lemma:reductions} and Theorem~\ref{thm:ll-contest-from-window-ts} imply:

\begin{corollary}
\label{cor:no wf sl counter from window}
If $n \geq 4$ and $3w \leq n$, there is no wait-free strongly linearizable implementation
of snapshot, max register, counter,
fetch\&inc-rement or fetch\&add that uses only $w$-window registers and test\&set objects.
\end{corollary}

\section{Discussion}

We introduced two variants of a contest object, one-shot and long-lived, which allow us to prove
impossibility results for lock-free and wait-free strongly linearizable concurrent objects,
respectively. Although these objects are strictly weaker than consensus, they suffice to capture
the coordination constraints imposed by strong linearizability.

Our results extend prior impossibility results in several directions. 
In particular, for lock-free implementations, 
they apply not only to read/write and window registers, 
but also to executions that use atomic stacks as base objects. 
Unlike interfering primitives, stack operations neither commute nor overwrite, 
yet stacks still do not suffice to support lock-free 
strongly linearizable implementations.
Moreover, unlike earlier reduction-based arguments, 
our proofs do not rely on readability assumptions.

Beyond the individual results, the contest objects help clarify why strong linearizability for some concurrent objects is
difficult to obtain. Informally, strong linearizability requires that once a competition among
concurrent operations is resolved, its outcome remains fixed in all extensions of the execution.
The results show that enforcing such consistent refereeing already requires substantial
coordination power, even when only a single distinguished process observes the outcome.

The proofs are based on valency-style arguments that apply 
across different progress conditions and classes of base objects. 
They indicate that strong linearizability enforces irrevocable
resolution of concurrency, 
even in settings where no agreement among processes is required.

Several directions remain open. One is whether contest objects, or variants thereof, can be used
as building blocks in strongly linearizable implementations that are lock-free or wait-free when
stronger primitives are available. Another is to extend Theorem~\ref{thm:ll-contest-from-window-ts} 
to window registers with $w \le n-1$.
We remark that the restriction $3w \le n$ (required for Lemma~\ref{lemma:to-reach-contr-wind-ts}(2)) 
stems from the need to maintain several disjoint sets of competitors 
whose actions cannot be simultaneously observed through a $w$-window register.
Finally, it would be interesting to identify a variant of the contest object that captures the impossibility
of lock-free strongly linearizable stacks from window registers and interfering primitives. Our results already show 
that such variant, if it exists, has to be strictly weaker than the contest object.

%

\bibliographystyle{plain}
\bibliography{references}

\end{document}